\documentclass[runningheads]{llncs}
\newif\ifproof
\prooftrue
\newif\ifconf
\conffalse
\usepackage[T1]{fontenc}
\usepackage{graphicx}
\usepackage{amsmath} %
\usepackage{amssymb} %
\usepackage{booktabs} %
\usepackage[normalem]{ulem}

\usepackage{xspace} %
\usepackage{stmaryrd} %
\usepackage{mdframed}

\usepackage[
  bookmarks=false,
  pdfpagelabels,
  bookmarksnumbered,
  bookmarksopen,
  bookmarksopenlevel=1,
  colorlinks=true,
  citecolor=black,
  filecolor=black,
  linkcolor=black,
  menucolor=black,
  urlcolor=black,
  pdftitle={Abstract Interpretation for Database Manipulating Processes},
  pdfauthor={Tobias Schüler, Stephan Mennicke, and Malte Lochau},
  pdfkeywords={database manipulating systems, abstract interpretation, labeled transition systems, bisimulation equivalence}]
  {hyperref}

\makeatletter
\newcommand{\xRightarrow}[2][]{\ext@arrow 0359\Rightarrowfill@{#1}{#2}}
\makeatother

\RequirePackage{xcolor}
\RequirePackage{marginnote}
\PassOptionsToPackage{dvipsnames,table,xcdraw}{xcolor}

\DeclareMathAlphabet{\mathcal}{OMS}{cmsy}{m}{n} %
\usepackage[scaled=.90]{helvet} %
\renewcommand*{\vec}[1]{\mathbf{#1}\xspace}
\usepackage{paralist}

\newcommand{\REVISED}[2][]{#2\xspace}

\newcommand\alphabet[1]{\ensuremath{\mathbf{#1}}\xspace}
\newcommand\constS{\alphabet{C}}
\newcommand\varS{\alphabet{V}}
\newcommand\predS{\alphabet{P}}
\newcommand\arity{\ensuremath{\textit{ar}}}
\newcommand\Inst[1][I]{\ensuremath{\mathcal{#1}}\xspace}
\newcommand\Jnst{\Inst[J]}
\newcommand\Dnst{\Inst[D]}
\newcommand\univ[2][U]{\ensuremath{\mathbb{#1}^{#2}}\xspace}
\newcommand\vars{\ensuremath{\textit{vars}}}
\newcommand\univi[2][I]{\ensuremath{\mathbb{#1}^{#2}}\xspace}
\newcommand\freevars{\ensuremath{\textit{free}}}

\newcommand\domain{\ensuremath{\textit{dom}}}
\newcommand\transrel[2][]{\ensuremath{\mathrel{\xRightarrow{#2}_{#1}}}}
\newcommand\tuple[1]{\ensuremath{\langle #1\rangle}\xspace}
\newcommand\firing[1]{\ensuremath{\mathrel{[#1\rangle}}}
\newcommand\enabled[1]{\ensuremath{[#1\rangle}\xspace}
\newcommand\act{\ensuremath{\textit{act}}\xspace}
\newcommand\actS{\ensuremath{\textsc{Act}}\xspace}
\newcommand\actsigma{\ensuremath{\tuple{\act,\sigma}}\xspace}
\newcommand\actsigmaS{\ensuremath{\actS \Sigma}\xspace}

\newcommand\setAdd{\ensuremath{\tt Add}}
\newcommand\setDel{\ensuremath{\tt Del}}

\newcommand\homeq{\ensuremath{\mathbin{\leftrightarrows}}\xspace}

\newcommand\conDom{\mathbb{C}\xspace}
\newcommand\absDom{\mathbb{A}\xspace}

\newcommand\nullS{\alphabet{N}\xspace}
\newcommand\nullE{\alphabet{n}\xspace}
\newcommand\guard{\ensuremath{g}\xspace}

\newcommand\A{P(\text{A})\xspace}
\newcommand\B{P(\text{B})\xspace}
\newcommand\C{P(\text{C})\xspace}

\newcommand\pX{P(x)}

\newcommand\pN[1]{\ensuremath{P(\nullE_{#1})}}

\newcommand\fNN[2]{\ensuremath{F(\nullE_{#1}, \nullE_{#2})}}
\newcommand\fXN[2]{\ensuremath{F(\text{#1}, \nullE_{#2})}}

\newcommand\fAB{F(\text{A},\text{B})}
\newcommand\fBA{F(\text{B},\text{A})}

\newcommand\fAC{F(\text{A},\text{C})}

\newcommand\fXY{F(x,y)}
\newcommand\fYX{F(y,x)}
\newcommand\fYZ{F(y,z)}
\newcommand\fZX{F(z,x)}

\begin{document}
\title{Abstract Domains for Database Manipulating Processes}
\titlerunning{Abstract Domains for DMS}
\author{
   Tobias Schüler\inst{1}\orcidID{0009-0008-1559-133X} \and
   Stephan Mennicke\inst{2}\orcidID{0000-0002-3293-2940} \and
   Malte Lochau\inst{1}\orcidID{0000-0002-8404-753X} 
}
\authorrunning{Schüler et al.}
\institute{
	University of Siegen, \email{firstname.lastname@uni-siegen.de} \and
	Knowledge-Based Systems Group, TU Dresden, Dresden, Germany \email{stephan.mennicke@tu-dresden.de}
}
\maketitle              %
\begin{abstract}
Database manipulating systems (DMS) formalize operations on relational databases 
like adding new tuples or deleting existing ones. 
To ensure sufficient expressiveness for capturing practical database systems, 
DMS operations incorporate as guarding expressions 
first-order formulas over countable value domains.
Those features impose infinite state, infinitely branching processes 
thus making automated reasoning about properties like reachability of states intractable. 
Most recent approaches therefore restrict DMS to obtain decidable fragments. 
Nevertheless, a comprehensive semantic framework capturing 
full DMS, yet incorporating effective notions 
of data abstraction and process equivalence is an open issue. 
In this paper, we propose DMS process semantics based 
on principles of abstract interpretation. 
The concrete domain consists of all valid databases, whereas the 
abstract domain employs different constructions for unifying sets of databases
being semantically equivalent up to particular fragments of the DMS guard language. 
The connection between abstract and concrete domain is effectively 
established by homomorphic mappings whose properties and restrictions 
depend on the expressiveness of the DMS fragment under consideration. 
We instantiate our framework for canonical DMS fragments 
and investigate semantical preservation of abstractions up to bisimilarity, 
being one of the strongest equivalence notions for operational process semantics. %
 
\end{abstract}
\section{Introduction}\label{sec:introduction}
\noindent\textbf{Background and Motivation.}
Modern software systems intensively interact with diverse environmental components which often includes
one or more (relational) databases.
Database manipulating systems~\cite{abdulla2018complexity} (DMS) and similar approaches~\cite{calvanese2015implementing,cangialosi2010conjunctive,bagheri2011foundations,montali2017db} characterize the operational behavior of (relational) database systems
by formalizing actions consecutively transforming the current state of 
databases by adding new tuples or deleting existing ones.
The action language supported by DMS-like formalisms must be sufficiently expressive 
to capture crucial behavioral aspects of practical database systems.
To this end, those actions combine set-based add/delete operations
with FOL formulas both defined on databases over (countable) value domains~\cite{abiteboul1995foundations}.
The FOL part serve as guarding expressions for actions
which, if enabled, have the ability to further expand and/or narrow the active domain of
databases reached in the subsequent state.
However, these distinct features of DMS-like formalisms impose intrinsically problematic
properties on the underlying operational semantics.
For instance, using labeled transitions systems (LTS)~\cite{calvanese2015implementing},
the resulting process model is not only non-regular and infinite-state, 
but even infinitely branching as arbitrary fresh data may be added to databases
 in a step.
Essential correctness properties of DMS processes like
reachability of states are thus not only theoretically undecidable, but
also practically intractable by state-of-the-art reasoning tools.
As a pragmatic workaround, most approaches consider bounded state spaces
and/or narrow down expressiveness of DMS languages to obtain decidable 
\REVISED[]{fragments~\cite{abdulla2018complexity}.}

\smallskip
\noindent\textbf{Contributions.}
In this paper, we apply the framework of \emph{abstract interpretation}~\cite{Cousot1977,dams1997abstract}
to tame the LTS semantics of DMS processes.
In the concrete domain, the set of LTS states corresponds to all valid
databases of a given database schema over infinite value domains.
In the abstract domain, LTS states are constructed by employing
different abstraction operators for unifying subsets of databases.
This abstract representation enables us to effectively connect
the abstract and concrete domains by means of homomorphic mappings.
The types of properties of DMS processes being preserved and/or reflected
by abstraction depend on the expressiveness of the DMS fragment used in DMS actions 
as well as the notion of process equivalence under consideration.
We instantiate our framework for canonical DMS fragments
and investigate behavior preservation of abstractions up to bisimilarity. 
As bisimilarity constitutes one of the strongest equivalence notions for LTS-based process semantics,
our abstraction builds the basis for guaranteeing preservation of essential semantical properties.
\REVISED[]{In this way, our framework provides a sound 
conceptual basis for building effective model-checking tools for DMS process verification~\cite{dams1997abstract}.}

\section{Foundations}\label{sec:foundation}
\paragraph{Databases.}
We assume a first-order (FO) vocabulary consisting of mutually disjoint (countably infinite) sets of constants $\constS$, variables $\varS$, and predicates $\predS$.
Each predicate $p\in\predS$ has an arity $\arity(p)\in\mathbb{N}$.
Terms are either constants or variables, and for a list of terms $\vec t = t_{1},\ldots,t_{n}$ we denote its length by $|\vec t|=n$.
An expression $p(\vec t)$ is an atom if $p\in\predS$ and $\vec t$ is a term list, such that $\arity(p)=|\vec t|$.
An atom is \emph{grounded} if it is variable-free and we call a finite set of ground atoms $\Dnst$ a \emph{database}.
The universe of all databases is $\mathbb{U}^{\constS}$.
A (possibly infinite) set of ground atoms $\Inst$ is an \emph{instance} with the respective universe $\univi{\constS}$.
Note that, $\univ{\constS} \subseteq \univi{\constS}$.

\paragraph{Guards.}
We consider FOL formulas $\REVISED[\varphi]{\guard}$ to serve as guards as follows:
\begin{equation}\label{eq:fo_formulas}
	\begin{array}{rcc|c|c|c|c}
  		\Phi & ::= & p(\vec {t})~&~t = u~&~\neg\Phi~&~\Phi\wedge\Phi~&~\exists x .\ \Phi 
	\end{array}
\end{equation}
where $p(\vec t)$ is an atom, $t,u$ are terms, and $x\in\varS$.
The terms occurring in guard $\REVISED[\varphi]{\guard}$ being variables are referred to by the set $\vars(\REVISED[\varphi]{\guard})$.
A variable $x\in\vars(\REVISED[\varphi]{\guard})$ is either \emph{free} or \emph{bound} in $\REVISED[\varphi]{\guard}$, defining the set $\freevars(\REVISED[\varphi]{\guard})$\footnote{$\freevars(r(\vec t)) = \vec t \cap \varS$, $\freevars(t = u) = \{ t, u \} \cap \varS$, $\freevars(\neg\REVISED[\varphi]{\guard}) = \freevars(\REVISED[\varphi]{\guard})$, $\freevars(\REVISED[\varphi]{\guard}\wedge\psi)=\freevars(\REVISED[\varphi]{\guard})\cup\freevars(\psi)$, and $\freevars(\exists x .\ \REVISED[\varphi]{\guard}) = \freevars(\REVISED[\varphi]{\guard})\setminus\{ x \}$.} of free variables.
\paragraph{Guards as FOL Fragments.}
\begin{table}[tbp]
	\caption{Guard fragments, their formula shape, and their abbreviation}
\label{fig:guard_classes}
\begin{tabular}{p{4.8cm}  p{1.5cm}  p{5.5cm} }
	\toprule
	guard fragment & 
	abbrv. & 
	formula \\ 
	\midrule
	\emph{normal conjuncitve guard} &
		NCG&
		$\exists \vec y .\ a_{1} \wedge \ldots \wedge a_{m} \wedge \neg b_{1} \wedge \ldots \wedge \neg b_{n}$ \\ 
	\emph{projection-free NCG} &
		pf-NCG&
		$a_{1} \wedge \ldots \wedge a_{m} \wedge \neg b_{1} \wedge \ldots \wedge \neg b_{n}$ \\ 
	\emph{conjuncitve guard}&
		CG&
		$\exists \vec y .\ a_{1} \wedge \ldots \wedge a_{m}$ \\ 
	\emph{projection-free CG} &
		pf-CG&
		$a_{1} \wedge \ldots \wedge a_{m}$ \\ 
	\emph{conjunction of negated atoms} &
		CNA&
		$\forall \vec y .\ \neg a_{1} \wedge \ldots \wedge \neg a_{m}$ \\ 
	\bottomrule
\end{tabular}
\end{table}
A \emph{normal conjunctive guard} (NCG) is a formula
\begin{equation}\label{eq:ncg}
	\exists \vec{y} .\ a_1 \wedge \ldots \wedge a_m \wedge \neg b_1 \wedge \ldots \wedge \neg b_n
\end{equation}
where $\vec{y}$ is a list of variables occurring in the atoms $a_1, \ldots, a_m, b_1, \ldots, b_n$.
For an NCG $\guard$ of shape \eqref{eq:ncg} we refer to 
the positive guard part by $\guard^{+} = \exists \vec y .\ a_{1} \wedge \ldots \wedge a_{m}$ and 
its negated part by $\guard^{-} = \exists \vec y .\ \neg b_{1} \wedge \ldots \wedge \neg b_{n}$, respectively.
Whenever convenient, $\guard^{+}$ ($\guard^{-}$, resp.) identifies the set of atoms occurring within $\guard$, meaning $\guard^{+} = \{ a_{1}, \ldots, a_{m} \}$ ($\guard^{-} = \{ b_{1}, \ldots, b_{n} \}$, resp.).
An NCG $\guard$ with $\guard^{-} = \emptyset$ is a \emph{conjunctive guard} (CG).
An NCG $g$ is \emph{safe} if $\vars(g^-)\subseteq\vars(g^+)$.
Similarly, the other guard fragments are summarized in Table~\ref{fig:guard_classes}.

\paragraph{Substitution.}
A \emph{substitution} is a partial function $\sigma : \varS \to \constS$ mapping variables to constants.
The set of all variables for which $\sigma$ is defined is denoted by $\domain(\sigma)$.
We call $\sigma$ a \emph{substitution for guard $\REVISED[\varphi]{\guard}$} if $\vars(\REVISED[\varphi]{\guard})\subseteq\domain(\sigma)$.
Such a substitution replaces variables of a guard by constants and, thereby, forms a \emph{guard match}.
For convenience, we assume for every substitution $\sigma$ and constant $c\in\constS$, $\sigma(c)=c$, extending the signature of $\sigma$ to $\varS \cup \constS \to \constS$.
If $\vec t = t_1 \dots t_n$ is a list of terms and $\sigma$ a substitution defined for all variables in $\vec t$, we denote by $\vec t\sigma$ the term list $\sigma(t_{1}) \dots \sigma(t_{n})$.
A substitution $\sigma$ is a \emph{match to guard $\REVISED[\varphi]{\guard}$ in instance $\Inst$} if (a) $\freevars(\REVISED[\varphi]{\guard})=\domain(\sigma)$ and (b) $\Inst, \sigma \models \REVISED[\varphi]{\guard}$, where
\begin{itemize}
  \item $\Inst, \sigma \models p(\vec t)$ if $p(\vec t\sigma)\in\Inst$,
  \item $\Inst, \sigma \models t = u$ if $t\sigma = u\sigma$,
  \item $\Inst, \sigma \models \neg\REVISED[\varphi]{\guard}$ if $\Inst, \sigma \models \REVISED[\varphi]{\guard}$ does not hold,
  \item $\Inst, \sigma \models \REVISED[\varphi]{\guard}\wedge\REVISED[\psi]{\guard'}$ if $\Inst,\sigma \models\REVISED[\varphi]{\guard}$ and $\Inst,\sigma \models \REVISED[\psi]{\guard'}$, and
  \item $\Inst, \sigma \models \exists x .\ \REVISED[\varphi]{\guard}$ if $\Inst, \sigma[x\mapsto c] \models \REVISED[\varphi]{\guard}$ for some $c\in\constS$.
\end{itemize}

\paragraph{Guard matches.} 
We denote the set of all matches to guard $\REVISED[\varphi]{\guard}$ in instance $\Inst$ by $\REVISED[\varphi]{\guard}(\Inst)$.
We may simply write $\guard$ to identify a guard.
A guard match to NCGs $\guard = \exists \vec y .\ \psi$ in $\Inst$ is tightly connected to the existence of homomorphisms from $\psi$ (viewed as a set of atoms) to instance $\Inst$.
A function $h : \constS\cup\varS \to \constS\cup\varS$ is called a \emph{homomorphism} from a set of atoms $\Inst[A]$ into a set of atoms $\Inst[B]$ if (a) $h(c)=c$ for all $c\in\constS$ and (b) $p(t_1, \dots, t_n)\in\Inst[A]$ implies $p(h(t_1), \dots ,h(t_n))\in\Inst[B]$.

\ifproof
	Also NCAs, CGs and their negation have a correspondence to homomorphisms as follows.
	\begin{proposition}\label{prop:hom-cg-lemma}
		For instance $\Inst$, CG $\guard$, and substitution $\sigma$, $\sigma\in \guard(\Inst)$ if, and only if, $\freevars(\guard)=\domain(\sigma)$ and
		there is a homomorphism $h : \guard \to \Inst$ such that $\sigma\subseteq h$.
	\end{proposition}
	\begin{proposition}\label{prop:hom-nca-lemma}
		For instance $\Inst$, NCA $\guard$, and substitution $\sigma$, $\sigma\in \guard(\Inst)$ if, and only if, $\freevars(\guard)=\domain(\sigma)$ and
		there is a function $h : \constS\cup\varS\to\constS$, such that (a) $h(c)=c$ for all $c\in\constS$, (b) $h(\guard)\cap\Inst = \emptyset$, and (c) $\sigma\subseteq h$.
	\end{proposition}
	This characterizations of guard matches to NCAs and CGs turns out to be quite useful in proofs of following sections.
\fi

\paragraph{Guard match, query answer and substitution.}
Guards and guard matches are very similar to queries and query answers in database systems.
However, whereas query answers should be \emph{domain independent} (i.e., having finitely many possible substitutions~\cite{abiteboul1995foundations}),
this does not necessarily hold for guard matches~\cite{abdulla2018complexity}.
For instance, query $\neg \pX$ would have infinitely many answers and is therefore prohibited, whereas
the corresponding guard simply checks if, for instance, a to-be-added person is not yet contained in the database.
\begin{example}\label{ex:guards}
	We consider a simplified social network (SSN)
	with two predicates,
	(1) a unary predicate $P(\textit{name})$ for \emph{persons} currently being members of the network with attributes $name$, and
	(2) a binary predicate $F(\textit{name}_1,\textit{name}_2)$ for a non-symmetric \emph{friendship} relation from person $\textit{name}_1$ to person $\textit{name}_2$. 
	We assume all possible strings denoting names to be part of \constS.
	A database of our SSN may be
	$\Dnst_e = \{\A, \B, \C, \fAB, \fBA, \fAC \}$ (with A, B and C may be Alice, Bob and Charles).
	Potential guards are
	\begin{itemize}
	\item a \underline{s}ymmetric \underline{f}riendship: $\guard_{sf} = \fXY \wedge \fYX$,
	\item a \underline{d}irected \underline{f}riendship: $\guard_{df} = \fXY \wedge \neg \fYX$,
	\item \underline{a} \underline{f}riendship from $x$ to someone: $\guard_{af} = \exists y.\fXY$, and
	\item \underline{n}o \underline{f}riendship: $\guard_{nf} = \neg \fXY \wedge \neg \fYX $.
	\end{itemize}
	On $\Dnst_e$ we obtain the following guard matches:
	\begin{itemize}
		\item $\guard_{sf}(\Dnst_e) = \{
				\{x\mapsto \text{A}, y \mapsto \text{B}\},
				\{x \mapsto \text{B}, y \mapsto \text{A}\}
			\}$,
			\item  $\guard_{df}(\Dnst_e) = \{
			\{x \mapsto \text{A}, y \mapsto \text{C}\}
		\}$,
		\item  $\guard_{sf}(\Dnst_e) = \{
			\{x \mapsto \text{A}\},
			\{x \mapsto \text{B}\}
		\}$, and
		\item  $\guard_{nf}(\Dnst_e) = \{
			\{x \mapsto \text{B}, y \mapsto \text{C}\},
			\{x \mapsto \text{C}, y \mapsto \text{B}\},
			\{x \mapsto \text{A}, y \mapsto \text{A}\},
			\dots
		\}$.
	\end{itemize}
	For instance $\Dnst_e, \{x \mapsto \text{A}, y \mapsto \text{B} \} \models \guard_{sf}$
	holds as Alice is a friend of Bob and Bob is a friend of Alice, whereas
	$\Dnst_e, \{x \mapsto \text{A}, y \mapsto \text{B} \} \models \guard_{df}$ does not hold.

\end{example}

\paragraph{Database Manipulating Systems.}
\REVISED[]{Database manipulating systems formalize 
possible sequences of \emph{actions} consecutively applied to database instances.
Syntactically, our formalization loosely follows the canonical notion
of actions used in the DMS formalism by Abdulla et al.~\cite{abdulla2018complexity}\xspace.}
An action consists of a \emph{guard} and an \emph{effect} on the current instance.
A guard specifies on which instances the action is applicable.
The effect might be deletion of atoms from the instance and adding new atoms to the instance.
Formally, the effect comprises two finite sets of atoms, $\setDel$ and $\setAdd$, such that $\vars(\setDel)\subseteq\freevars(\REVISED[\varphi]{\guard})$.
Atoms in $\setDel$ are determined by the match for guard $\guard$, while $\setAdd$ is a collection of new atoms.
Note that $\setDel$ and $\setAdd$ may contain variables that will be bound by (a) the guard matches and (b) by arbitrary constants in case of those variables in $\vars(\setAdd)\setminus\freevars(\REVISED[\varphi]{\guard})$.
The rationale behind case (b) is that an action inserting atoms may depend on external stimuli like sensor data or user input.
An action $\act$ is a triple $(\guard, \setDel, \setAdd)$ which forms the basis of a \emph{database manipulating system} (DMS).
\begin{definition}[Database Manipulating System]\label{def:dms}
A \emph{database manipulating system} (DMS) is a pair $\mathcal{S} = (\Inst_{0}, \actS)$ where $\Inst_{0}$ is the \emph{initial instance} and $\actS$ is a finite set of actions.
\end{definition}
From $\Inst_{0}$, \REVISED[]{any sequence} of actions $\act = (\guard,\setDel,\setAdd)\in\actS$ may be performed 
based on substitutions $\sigma$ due to matches of guard $\guard$.
Note that $\sigma$ specifies all variables occurring in $\setDel$.
We denote by $\setDel\sigma$ the set obtained by replacing all occurrences of variables $x\in\vars(\setDel)$ by $\sigma(x)$.
\REVISED{In general, for a set of atoms $\mathcal{A}$ and substitution $\sigma$, $\mathcal{A}\sigma$ is the set of atoms in which each variable $x\in\vars(\mathcal{A})$ has been replaced by $\sigma(x)$ if it is defined for $\sigma$.}
Set $\setAdd$ may contain variables which are not in $\domain(\sigma)$
such that $\setAdd\sigma$ is not a proper database \REVISED[or instance]{}.
To facilitate arbitrary external inputs, we expand $\sigma$ to the missing variables.
Substitution $\sigma^{\star}$ extends $\sigma$ to $\setAdd$ if $\domain(\sigma^{\star})=\setAdd$ and $\sigma\subseteq\sigma^{\star}$.
\REVISED{Extending $\sigma$ to $\sigma^{\star}$ completes a step by deletions $\setDel\sigma^{\star}$ from and additions $\setAdd\sigma^{\star}$ to the current instance.}
\begin{definition}[DMS Step]\label{def:dms-step}
 A DMS action $\act = (\guard,\setDel,\setAdd)$ is \emph{enabled under instance $\Inst$ and substitution $\sigma$}, denoted $\Inst\enabled{\act,\sigma}$, if $\sigma\in \guard(\Inst)$.
 If $\Inst\enabled{\act,\sigma}$, then an effect is an extension $\sigma^{\star}$ of $\sigma$ to $\setAdd$, producing instance $\Inst' = (\Inst \setminus \setDel\sigma^{\star}) \cup \setAdd\sigma^{\star}$.
  We denote the \emph{DMS step} from $\Inst$ to $\Inst'$ via $\act$ and $\sigma$ by $\Inst\firing{\act,\sigma^{\star}}\Inst'$.
\end{definition}

\begin{example}
\label{ex:dms_action}
	Action $\act_{\textit{add}} = (\textit{true}, \emptyset, \{\pX\})$ (adding a new person)
	is enabled under each instance even if the person already exists in that instance.
	Thus, $\sigma$ is empty and $\sigma^{\star}$ may be, e.g., $\{x \mapsto \text{A}\}$.
	Action $\act_{\textit{rev}} = (\guard_{df},  \{\fXY\}, \{\fYX\})$
	checks if a directed friendship exists between $x$ and $y$, 
	deletes this friendship and adds the reversed friendship.
\end{example}
The formal semantics of a DMS is defined as a \emph{labeled transition system} (LTS).
\begin{definition}[Labeled Transition System]\label{def:lts}
A \emph{labeled transition system} (LTS) is a triple $\mathcal{T}=(Q,\Sigma,\transrel{})$ where $q$ is a set of states (processes), $\Sigma$ is a set of transition labels, and $\transrel{} \subseteq Q \times \Sigma \times Q$ a transition relation.
We denote $(q,a,q')\in\transrel{}$ as $q \transrel{a} q'$ and write $q \transrel{a}$ if $\exists q'\in Q : q \transrel{a} q'$
and $q\not\transrel{a}$ if not $q\transrel{a}$.
\end{definition}
LTS $\mathcal{T} = (Q,\Sigma,\transrel{})$ is
\begin{inparaenum}[(a)]
	\item \emph{finitely branching} if for each $q\in Q$, the set $\{ q'\in Q \mid \exists a \in \Sigma : q \transrel{a} q' \}$ is finite,
  \item \emph{image-finite} if for each $q\in Q$ and $a\in \Sigma$, the set $\{ q'\in Q \mid q \transrel{a} q' \}$ is finite,
	\item \emph{finite-state} if $Q$ is finite, and
	\item \emph{deterministic} if for each state $q\in Q$ and $a\in\Sigma$, $q\transrel{a} q'$ and $q\transrel{a} q''$ implies $q'=q''$.
\end{inparaenum}
Although LTSs may be directly associated with directed edge-labeled graphs, comparison relations
based on graph homomorphisms are too strong to capture distinctive features of LTS processes.
Instead, \emph{simulation} and \emph{bisimulation} relations on processes are used.
Intuitively, process $q$ simulates $p$ if every action that may be performed by $p$
can be mimicked by $q$ and the successor states again simulate each other.

\begin{definition}[(Bi-)Simulation]\label{def:sim}
For an LTS $(Q,\Sigma,\transrel{})$, a binary relation $R\subseteq Q\times Q$ is 
a \emph{simulation} if for all $(p,q)\in R$ and $a\in\Sigma$, $p\transrel{a} p'$ implies that $q'\in Q$ exists such that $q\transrel{a} q'$ and $(p',q')\in R$.
Process $q\in Q$ \emph{simulates} process $p\in Q$ if there is a simulation $R$ with $(p,q)\in R$.
If $p$ simulates $q$ \REVISED{by simulation $R$}, and $q$ simulates $p$ \REVISED{by simulation $R'$}, then $p$ and $q$ are \emph{similar}.
Simulation $R$ is a \emph{bisimulation} if, and only if, $R^{-1} := \{ (q,p) \mid (p,q)\in R \}$ is also a simulation.
If there is a bisimulation $R$, such that $(p,q)\in R$, then $p$ and $q$ are \emph{bisimilar}.
\end{definition}
\REVISED{Note, the witnesses $R$ and $R'$ for similarity are not necessarily bisimulations as possibly $R^{-1}\neq R'$.}

DMS semantics can be formalized as an LTS 
	$\textsf{DMS} := (\univ{\constS}, \actsigmaS, \transrel{})$
where $\transrel{}$ is formed by $\Inst_{1} \transrel{\actsigma} \Inst_{2}$ if, and only if, $\Inst_{1} \firing{\act,\sigma} \Inst_{2}$ (cf.\xspace Def.~\ref{def:dms-step}).
In general, $\textsf{DMS}$ is infinitely branching, infinite-state, and deterministic. \REVISED[, and image-finite.]{}

$\textsf{DMS}$ builds the basis for investigating
desirable properties of all possible processes defining a DMS.
For instance, the \emph{reachability problem} 
asks for a given $\textsf{DMS}$ and a distinguished action $\act_x$,
if there is an instance $\Inst_{x}$ with $\Inst_{0} \transrel{} \Inst_{1} \transrel{} \dots \transrel{} \Inst_{x}$ 
such that action $\act_x$ is enabled under $\Inst_{x}$.
The reachability problem is undecidable for DMS~\cite{abdulla2018complexity}\xspace.
\ifproof
	\begin{example}\label{example:reachability}
		We introduce a new binary predicate $W(w_1,w_2)$ with $w_1$ and $w_2$ being words 
		over $\constS^*$ where $W(\epsilon,\epsilon)$ holds for the initial database.
		For each friendship predicate $\fXY$ we consider a set of DMS actions of the form
		$\act_{\fXY} := (W(w_1,w_2), W(w_1,w_2), W(w_1 \circ x, w_2 \circ y))$
		where $\circ$ denotes concatenation.
		We consider the DMS action $\act_{\textit{end}} = (W(w_1,w_2) \wedge w_1 = w_2, \emptyset, \emptyset)$.
		Then, the reachability problem with respect to $\act_{\textit{end}}$
		is undecidable as it can be reduced to Post's correspondence problem.
	\end{example}
\fi
The next example constitutes a semi-decidable reachability problem.
\begin{example}\label{example:semi-reachability}	
Given a predefined set of persons and actions
for consecutively adding and deleting friendships between arbitrary pairs of persons,
do we eventually reach a database containing a triangle friendship between three different persons
(i.e., $x$ a friend of $y$, $y$ a friend of $z$ and $z$ a friend of $x$)?
To this end, we expand the unary predicate $P(\textit{name})$ to a binary 
predicate $P(\textit{name},\textit{name})$ and use 
NCG to define a guard $P(x,x) \wedge P(y,y) \wedge \neg P(x,y)$. $\neg P(x,y)$
(i.e., ensuring that $x$ and $y$ match different persons).
This is a standard technique to avoid $\neq$ in first order formulas.
Starting from an arbitrary database, we consider two actions:
$\act_{add} := ( P(x,x) \wedge P(y,y) \wedge \neg P(x,y) , \emptyset,\fXY)$ 
(adding a friendship) and $\act_{delete} := (\fXY, \fXY, \emptyset)$
(deleting a friendship) and ask for reachability of the action 
$\act_{end} = (\exists x,y,z. \fXY \wedge \fYZ \wedge \fZX \wedge \neg P(x,y) \wedge \neg P(y,z) \wedge \neg P(z,x), \emptyset, \emptyset)$.
\end{example}
A finite solution to this problem comprises an \emph{abstract LTS} with four states, where each of those abstract states contains all subsets of databases with 
(1) no friendships, (2) friendship chains 
of maximum length $\leq 2$, (3) friendship chains of maximum length $> 2$ without any triangles, and (4) at least one triangle.

In the remainder of this paper, we develop a hierarchy of abstract domains
to characterize semantic-preserving abstractions of states of \textsf{DMS} depending
on the expressiveness of the guard fragment used.
Our approach is based on the formal framework of abstract interpretation.

\section{Principles of Abstract Interpretation}\label{sec:abst-intpre}
Before we present our abstract interpretation framework for DMS, we first describe its basic ingredients.
Different processes assembled in $\mathsf{DMS}$ may share similar behavior in terms
of their enabled actions and subsequent processes.
For instance, let us consider a DMS action which inserts a friendship between
Alice and Bob, where the guard of this action consists of a conjunction of atoms requiring 
Alice and Bob to exist in the database.
All (i.e., countably infinitely many) \emph{concrete states} matching this guard may be aggregated 
into \emph{one} single abstract state. 
The concrete states aggregated in the subsequent abstract state
reached after performing this action then all share 
the inserted relationship between Alice and Bob.
The way how the concrete states are aggregated into,
and reconstruction from, such an abstract state clearly depends
on the guard fragment used.
In addition, DMS states are infinitely branching due to
the ability of DMS actions to insert any possible new value.
However, in many cases, the exact values are often not relevant
for reasoning about the subsequent behavior and can therefore 
be aggregated into one representative abstract value.
The following definitions are based on Dams et al.~\cite{dams1997abstract} and 
conceptualize these observations.

\paragraph{Lattice.}
Abstract interpretation provides a framework
for effectively reasoning about computational models over infinite semantic domains modeled as lattices.
By $\sqcap$ and $\sqcup$ we denote binary operations on sets $S$.
The operators $\sqcap$ and $\sqcup$ are monotone with respect to a partial order $\leq$ on $S$ 
(i.e., $x_1,x_2,y_1,y_2 \in S$, $x_1 \leq x_2$ and $y_1 \leq y_2$ implies  
$x_1 \sqcap y_1 \leq x_2 \sqcap y_2$ and  $x_1 \sqcup y_1 \leq x_2 \sqcup y_2$).
\begin{definition}[Lattice]\label{def:lattice}
	A \emph{lattice} is a partially ordered set $(S, \leq)$ such that each two-element subset $\{x,y\} \subseteq S$ has
	(1) a unique least upper bound in $S$, denoted by $x \sqcup y$, and
	(2) a unique greatest lower bound in $S$, denoted by $x \sqcap y$.
\REVISED[A \emph{bounded lattice} has a unique greatest element $\top$ and a unique least element $\bot$ such that
$\top \geq \bigsqcup_{x \in S} x$, $\bot \leq \bigsqcap_{x \in S} x$ and
$\bot \leq x \leq \top$ for all $x \in S$. 
$\top$ and $\bot$ are identities for $\sqcap$ and $\sqcup$.]{}
\ifproof
A \emph{bounded lattice} has a unique greatest element $\top$ and a unique least element $\bot$ such that
$\top \geq \bigsqcup_{x \in S} x$, $\bot \leq \bigsqcap_{x \in S} x$ and
$\bot \leq x \leq \top$ for all $x \in S$. 
$\top$ and $\bot$ are identities for $\sqcap$ and $\sqcup$.
\fi
\end{definition}
Bounded lattices are not further deployed in the following 
but are mentioned here only for the sake of comprehensibility.
Abstract interpretation aims at 
establishing connections between lattices modeling different semantic domains.

\paragraph{Galois Connection.}
By $(\conDom,\sqsubseteq)$ we denote a concrete semantic domain where
$\conDom=2^{Q}$ comprises the set of all subsets of concrete sets of states $Q$ 
of a computational model (here: $\mathsf{DMS}$).
By $\sqsubseteq$ we denotes a partial (semantic) ordering on $\conDom$ (here: $\subseteq$).
By $(\absDom,\preceq)$ we denote an abstract semantic domain
where $\absDom$ is a set of abstract states and $\preceq$ a partial (precision) ordering on $\absDom$.
It is crucial that the elements of the concrete domain $\conDom$ are possible \emph{subsets} of concrete states, 
whereas the elements of the abstract domain $\absDom$ are \emph{singleton} abstract states.
The mutual connection between concrete and abstract domain
is shaped by a pair of abstraction function $\alpha : \conDom \to \absDom$ and a concretization function 
$\gamma : \absDom \to \conDom$, together forming a \emph{Galois connection}.
\begin{definition}[Galois Connection]\label{def:absint}
	The pair $(\alpha : \conDom \rightarrow \absDom, \gamma : \absDom \rightarrow \conDom)$ is a \emph{Galois connection}
	between lattices $(\conDom,\sqsubseteq)$ and $(\absDom,\preceq)$ if
	(1) $\alpha$ and $\gamma$ are \emph{total} and \emph{monotone},
	(2) $\forall C\in \conDom:\gamma\circ\alpha(C)\sqsupseteq C$, and
	(3) $\forall a\in \absDom:\alpha\circ\gamma(a)\preceq a$.
\end{definition}
Monotonicity guarantees that more precise abstractions 
single out fewer concrete states and, conversely, abstracting larger sets of concrete states 
yields less precise abstractions.
Furthermore, (2) requires that concrete states are preserved after reconstruction.
Finally, (3) requires a form of optimality of the abstraction thus not decreasing precision. 

\ifproof
\paragraph{Bisimulation.}
We lift (bi-)simulations to steps $C \transrel{a} C'$ between sets of concrete states 
as apparent in the concrete domain in two ways: either (1) there is at least one process $q\in C$ with $q\transrel{a} q'$ and $q'\in C'$,
or (2) all processes in $C$ evolve to some process in $C'$ by action $a$.
We may refer to case (1) as an $\exists$-step and to case (2) as a $\forall$-step,
and adapt the notions of (bi-)simulation to the these new types of steps.
\begin{definition}[$\forall$/$\exists$-(bi-)simulation]\label{def:bisim}\label{def:forall-sim}
	For abstract domain $(\absDom,\preceq)$ and concrete domain $(\conDom,\sqsubseteq)$,
	a binary relation $R \subseteq \absDom \times \conDom$ 
	is an \emph{$\exists$-simulation} if for all $(\Inst[A], C) \in R$ and $a \in \actS$,
			$\Inst[A] \transrel{a} \Inst[A]'$ implies that there is a $C'\in \conDom$ such that (a) $C \transrel{a} C'$ with a $q\in C$ for each $q'\in C'$ such that $q\transrel{a} q'$, and (b) $(\Inst[A]',C')\in R$.
	$R$ is a \emph{$\forall$-simulation} if, additionally to being an $\exists$-simulation,
	for each $q\in C$ there is a $q'\in C'$ with $q\transrel{a} q'$.
\end{definition}
If $(\Inst[A],C)\in R$ and $R$ is a $\forall$/$\exists$-simulation, we say that $C$ $\forall$/$\exists$-simulates $\Inst[A]$.
By reversing the conditions of $\forall$/$\exists$-simulations, we get $\forall$/$\exists$-simulations between the concrete domain and the abstract domain (i.e., $R\subseteq\conDom\times\absDom$).
Naturally, a $\forall$/$\exists$-simulation $R$ is called a $\forall$/$\exists$-bisimulation if, and only if, $R^{-1}$ is a $\forall$/$\exists$-simulation.
In this paper we restrict our considerations to $\forall$-bisimulations
which preserves reachability properties as in our running example.
Analogously,  $\exists$-bisimulation preserve safety properties and will be considered as a future work.

Please note that the different system types (single instances vs.\xspace sets) introduces a slight asymmetry into the notions of $\forall$/$\exists$-(bi-)simulations.
Although every $\forall$-simulation $R\subseteq\absDom\times\conDom$ is an $\exists$-simulation, a $\forall$-simulation $R\subseteq\conDom\times\absDom$ does not necessarily adhere to the requirements of an $\exists$-simulation.
The reason is that $(C,\Inst[A])$ of a $\forall$-simulation $R\subseteq\conDom\times\absDom$ considers only steps $C \transrel{a} C'$ that are complete (i.e., every database in $C$ evolves) while if $R$ is viewed as an $\exists$-simulation, only one database from $C$ may evolve, say to some singleton set $C''$ which is not necessarily captured by $R$.
Therefore, the $\forall$-bisimulation results we obtain throughout the next section do not entail respective $\exists$-bisimulation results.
\fi

\paragraph{Abstract Interpretation Framework.}
The remainder of this paper is devoted to a hierarchy of concrete domains $(\mathbf{2}^{\univ{}}, \subseteq)$
for \textsf{DMS} processes shaped by different fragments of FOL \REVISED[in the guard]{as guard language}, where
the functions $\gamma$ and $\alpha$ are either based on the 
supremum or infimum of the corresponding abstract domains.
\REVISED{For a guard language $\mathcal{L}$, we call a Galois connection $(\alpha,\gamma)$ an \emph{abstract interpretation w.r.t.\xspace $\mathcal{L}$} if for each set of databases $C\subseteq\univ{}$ and set of DMS actions $\actS$ using only guards from $\mathcal{L}$, $\alpha(C)$ and $C$ are $\forall$-bisimilar.}
\section{Abstract Interpretation of DMS}\label{sec:abstract_interpretation_of_DMS}

The concrete domain is fixed: $(\mathbf{2}^{\univ{}}, \subseteq)$.
For (possibly infinite) sets $C$ of databases, we effectively present six different abstractions:
The first two very basic ones are based on (set) union and intersection.
The third abstraction is a (Cartesian) combination of the two prior abstractions with the benefit of supporting a more practical guard fragment.
One caveat \REVISED[to mention]{} about these abstractions \REVISED{is that we have} to waive projection (i.e., existential quantification).
To gain DMS actions with more expressive guards, and thereby capture more realistic systems, we devise abstractions for more general abstract domains.  
We expand our abstract domain incorporating so-called \emph{labeled nulls} as terms in abstract instances.
The order on the abstract domain is then based on homomorphisms.
The three remaining abstractions are complements of the first three, now in the more abstract domain incorporating labeled nulls.
Table~\ref{fig:abstract_domains_and_guard_languages} summarizes our results.
\begin{table}[tbp]
	\caption{Abstract Domains, Interpretations, and Respective Guard Fragments}\label{fig:abstract_domains_and_guard_languages}
\centering
\begin{tabular}{ lllrr }
	\toprule
	abstract domain ~~~~& 
	$\alpha(C)$~~~~~~~~~~~~~~~~~~& 
	$\gamma(\Inst)$~~~~~~~~~~~~~~~~~~~~& 
	~~~~fragment &
	~~~~~~~~~~~~~~~~~~\\ 
	\midrule 
		$(\univi{} ,\subseteq)$ &
		$\bigcup C$ &
		$  \{ \Dnst \subseteq \Inst \}$ &  
		CNA &
		Theorem~\ref{theorem:union_abstraction_forall} \\ 
		$(\univi{} ,\supseteq)$  & 
		$\bigcap C$ & 
		$ \{ \Inst \subseteq \Dnst\}\ $ & 
		pf-CG&
		Theorem~\ref{theorem:intersection_abstraction_forall}\\
		$(\univi{} \times \univi{},\leq)$  & 
		$(\bigcap C, \bigcup C)$ & 
		$ \{ \Inst \subseteq \Dnst \wedge \Dnst \subseteq \Inst \}\ $ & 
		pf-NCG &
		Theorem~\ref{theorem:combination_abstraction_forall}\\
		$(\univi{\nullS} ,\rightarrow)$ & 
		$\bigsqcup C$ & 
		$\{\Dnst \to \Inst\}$ & 
		CNA&
		Theorem~\ref{theorem:union_abstraction_null_forall} \\
		$(\univi{\nullS} ,\leftarrow)$ & 
		$\bigsqcap C$ & 
		$\{\Inst \to \Dnst\}$ &
		CG &
		Theorem~\ref{theorem:intersection_abstraction_null_forall} \\
		$(\univi{\nullS} \times ,\univi{\nullS}, \preceq)$ & 
		$(\bigsqcap C, \bigsqcup C)$ & 
		$\{\Inst \to \Dnst \wedge \Dnst \to \Inst\}$ &
		NCG &
		Theorem~\ref{theorem:combination_abstraction_null_forall}\\
		\bottomrule
\end{tabular}
\end{table}

The rest of this section is structured as follows.
First, we introduce a naive set-based abstraction based on the set union operator on databases together with a summary of further set-based abstractions.
Resolving the issue of neglecting variable projections in guards we introduce instances with labeled nulls, on which CGs can be used without losing precision.
Finally, we combine unions and intersections to even support DMS actions with NCGs.
\ifproof
Other abstractions are mentioned. All proofs are attached in the appendix.
\fi

\subsection{Set-Based Abstractions: The Case of Union}
As a first and very basic abstraction we study $\bigcup C$ of any set $C\in\mathbf{2}^{\univ{}}$ of databases.
If $C$ is infinite, $\bigcup C$ is infinite as well, meaning that $\bigcup C$ is captured in $\univi{}$.
Henceforth, we facilitate $\bigcup C$ via the abstraction function $\alpha_{1} : \mathbf{2}^{\univ{}} \to \univi{}$ with 
$\alpha_{1}(C)$.
\begin{align}
	\alpha_{1}(C) := \bigcup C
	&  & 
	\gamma_{1}(\Inst) := \{ \Dnst \subseteq \Inst \mid \text{$\Dnst$ is a database} \} \label{eq:alphagamma_union}
\end{align}
The natural choice for the abstract domain is, thus, $(\univi{}, \subseteq)$ because the more databases $C$ contains, the bigger the abstract instance is (cf.\xspace Def.~\ref{def:absint} item 1).
The counterpart concretization function $\gamma_{1} : \univi{} \to \mathbf{2}^{\univ{}}$ is determined by $\alpha_{1}$: While $\alpha_{1}$ forms the union of all databases contained in a set of databases $C$, an abstract instance then describes all databases that are (finite) subsets of the abstract instance. $\gamma_{1}(\Inst)$ is defined in (\ref{eq:alphagamma_union}).

Databases are finite by definition, implying that if $\Inst$ is infinite, $\Dnst\subsetneq\Inst$ for every $\Dnst\in\gamma_{1}(\Inst)$.
The functions in (\ref{eq:alphagamma_union}) make up for a Galois connection.
\begin{proposition}\label{theorem:union_galois_connection}
  $(\alpha_{1},\gamma_{1})$ is a Galois connection.
\end{proposition}
For $C\in\mathbf{2}^{\univ{}}$, we are interested in the behavioral properties of the abstraction $\alpha_{1}(C)$.
Therefore, observe that for every database $\Dnst\in C$, $\Dnst \subseteq \alpha_{1}(C)$.
Thus, guards asking for the absence of atoms will have the same matches on all the databases in $C$ as well as the abstraction $\alpha_{1}(C)$.
\begin{example}
	We analyze two guards $\guard_{nf}$ (absence of a friendship) and $\guard_{sf}$ (presence of a symmetric friendship) from example~\ref{ex:guards} on $C$ and $\Inst$ with	
	$C =\{ 
		\{\A,$ $\B,\fAB,\fBA\}, 
		\{\A,\B,\C\}
	\}$ and
	$\Inst = \bigcup C = \{\A,\B,$ $\C, \fAB, \fBA\}$.	
	If a friendship is absent in each database of $C$, this friendship is also absence in $\Inst$ (i.e., the union of all databases of $C$). If a friendship is absent in $\Inst$ this friendship is also absent in each database of $C$.
	In contrast, the presence of a symmetric friendship like $\fAB, \fBA$ holds for $\Inst$ but not for each database in $C$.

	The guard $\guard_{nf} = \forall y. \neg \fXY$ ensures the absence of all friendships of a person $x$ through the universal quantifier. $\guard_{nf}$ behaves similar to $\guard_{nf}$.
	The behavior of the existential quantifier is conversely.
	For instance,   $\guard_{ex} =  \exists x. \neg \pX$ holds for each database as databases are finite but the set of all constants is infinite.
	In contrast, if set $C$ is infinite and for each constant $c$, $P(c)$ is contained in some database in $C$,
	$\Inst = \bigcup C$ does not satisfy $\guard_{ex}$.
\end{example}
As the examples show, $\alpha_{1}(C)$ may enable DMS actions with conjunctive guards that are not enabled by some, or any, of the concrete databases in $C$. %
Thus, $\alpha_{1}(C)$ captures the behavior of all databases in $C$ if we choose CNA guards.
\begin{theorem}\label{theorem:union_abstraction_forall}
	\REVISED[The Galois connection ]{}$(\alpha_{1},\gamma_{1})$ is an abstract interpretation w.r.t.\xspace CNA guards\REVISED[ and $\forall$-bisimilarity]{}.
\end{theorem}	
Similarly, we obtain an abstraction framework based on intersection of all the databases contained in set $C$ of concrete databases.
\begin{theorem}\label{theorem:intersection_abstraction_forall}
	\REVISED[The ]{}Galois connection $(\alpha_{2},\gamma_{2})$ with $\alpha_2(C) := \bigcap C$ and $\gamma_2(\Inst) := \{ \Dnst\in\univ{} \mid \Inst \subseteq \Dnst  \}$ is an abstract interpretation for pf-CGs\REVISED[ and $\forall$-bisimilarity]{}.
\end{theorem}
This is a special case of Theorem~\ref{theorem:intersection_abstraction_null_forall} (cf.\xspace next subsection). 
Furthermore, combining both former abstractions allows us to cover projection-free normal conjunctive guards in DMS actions.
The rationale behind this abstraction is that for an NCG $g$, $g^+$ is evaluated on the intersection component while $g^-$ is simultaneaously evaluated on the union component of the abstraction.
\begin{theorem}\label{theorem:combination_abstraction_forall}
	For $\alpha_3(C) := (\alpha_1(C),\alpha_2(C))$ and $\gamma_3((\Inst^\cup,\Inst^\cap)) := \{ \Dnst\in\univ{} \mid \Inst^\cap \subseteq \Dnst \subseteq \Inst^\cup  \}$, Galois connection $(\alpha_{3},\gamma_{3})$ is an abstract interpretation for pf-NCGs.
\end{theorem}
Next, we consider abstractions allowing for projections (i.e., existentially quantified variables in DMS action guards) to fully capture NCGs in DMS actions.

\subsection{Abstractions with Labeled Nulls: The Case of Intersection}\label{sec:abst_null}
There are two issues with the abstractions discussed so far:
(a) limited expressiveness in guards of DMS actions (no existential quantification) and
(b) (still) infinite branching of abstract states.
The reason for the latter is that abstract instances resemble their concrete counterparts too explicitly.
To resolve both issues we use the well-known \emph{labeled null} abstraction to get a notion of existence of values contained in a database whose exact values are irrelevant.
Finite branching is a welcome side-effect of this abstraction as well as a precise abstraction for DMSs using CGs (including projection via existential quantification).

\emph{Labeled nulls} are introduced in our framework as a countably infinite set $\nullS$ (disjoint from all other term sets).
As labeled nulls are proxies for the existence of values (i.e., constants), a database, in which every occurrence of a null is replaced by a constant (or other null), is certainly related to the instance that uses the null.
Let us denote the set of all instances using constants and labeled nulls by $\univi{\nullS}$ (short for $\univi{\constS \cup \nullS}_{\predS}$).
The notions of homomorphisms and guard matches naturally extend to databases containing nulls (i.e., constants must still map to constants, but nulls may map to nulls or constants).
Due to the nature of labeled nulls, their identity does not have the same role as constants have.
It is natural to consider $\univi{\nullS}$ closed under equivalence up to homomorphisms.
This means, instances $\Inst[I],\Inst[J]\in\univi{\nullS}$ are equal, denoted $\Inst[I]\homeq\Inst[J]$, if $\Inst[I]\to\Inst[J]$ and $\Inst[J]\to\Inst[I]$.
Note, on $\univ{}$ equivalence up to homomorphisms coincides with set equality.
For instance $\{\A\} \leftrightarrows \{\A,\pN{0}\}$ because we can map $A$ on $A$ and $\nullE_0$ on $A$.
$\{\fNN{0}{1}\} \to \{\fNN{0}{0}\}$ but $\{\fNN{0}{0}\} \not\to \{\fNN{0}{1}\}$ because we can not map $\nullE_0$ on $\nullE_0$ and $\nullE_0$ on $\nullE_1$.

$(\univi{\nullS},\to)$ forms a lattice and, by duality, $(\univi{\nullS},\leftarrow)$, too.
The join $\sqcup$ of $(\univi{\nullS},\to)$ is simply the union of the instances.
Conversely, $\sqcap$ is an intersection of two instances generalizing common atoms with different constants via null assertions.
For instance, $\Inst[I] = \{\A,\B, \fNN{0}{1}\}$ and $\Inst[J] = \{\A,\C,\fAC\}$ have
$\Inst[I] \sqcup \Inst[J] = \{\A,\B,\C,\fNN{0}{1},$ $\fAC\}$ as least upper bound
and the  greatest lower bound is 
$\Inst[I] \sqcap \Inst[J] = \{\A,\fNN{0}{1} \}$.

The next two definitions describe how an action is performed in $(\univi{\nullS},\to)$.
Let $\Inst$ be an instance and $\act = (g,\setDel,\setAdd)$ a DMS action.
Instead of extending guard matches $\sigma$ to $\sigma^{\star}$ (involving some constants that are added to the instance through variables in $\setAdd$), we consider extensions of $\sigma$ that insert (globally) fresh labeled nulls for all variables in $\vars(\setAdd)\setminus\freevars(\guard)$.
\begin{definition}
	\label{def:shortened-step}
	Let $\act = (\guard,\setDel,\setAdd)$ be a DMS action.
	For abstract instance $\Inst\in\univi{\nullS}$, if $\sigma\in \guard(\Inst)$, then $\Inst\transrel{\actsigma}(\Inst\setminus\setDel\sigma^{\star})\cup \setAdd\sigma^{\star}$ where $\sigma\subseteq\sigma^{\star}$ and for each variable $x\in\vars(\setAdd)\setminus\freevars(\guard)$, $\sigma^{\star}(x)$ is a fresh labeled null. %
\end{definition}
\begin{example}
\label{ex:shortened-step}
\REVISED[]{For action $\act_{\textit{add}} = (\textit{true}, \emptyset, \{\pX\})$ from example~\ref{ex:dms_action},
$\vars(\setAdd)\setminus\freevars(\guard) = \{ \mathit{x} \} \setminus \emptyset = \{ \mathit{x} \}$
and $\sigma^{\star}(x) = \nullE$.
We obtain $\emptyset\transrel{\tuple{\act_{\textit{add}} ,\emptyset}} \{\pN{}\}$.}
\end{example}
Note that the action label only contains the match $\sigma$ and not its extension.
The reason is that for instances $\Inst,\Inst[B]_1,\Inst[B]_2$ and action-match pair $\actsigma$, if $\Inst \transrel{\actsigma} \Inst[B]_1$ and $\Inst \transrel{\actsigma} \Inst[B]_2$, then $\Inst[B]_1 \leftrightarrows \Inst[B]_2$.
Thus, the different target instances cannot be distinguished in our abstract domain.
This notion of steps is similar to what the Chase does in existential rule reasoning~\cite{fagin2005data}.
Due to the closure of the domain under homomorphisms, it also resembles the standard chase and the core chase to certain extents~\cite{deutsch2008chase}.
Sets of concrete instances still proceed as originally defined in Sect.~\ref{sec:foundation}.
To still guarantee a resemblance between the action labels in our abstract domain and the labels used for concrete instances (where no nulls are involved), we introduce a notion of \emph{compatibility} \REVISED[between]{of} action labels.
\begin{definition}
	\label{def:compatibility}
	\REVISED[Two action labels]{Action label} $\langle \act_1, \sigma_1 \rangle$ \REVISED[and]{is \emph{compatible to} action label} $\langle \act_2, \sigma_2 \rangle$, denoted by $\langle \act_1, \sigma_1 \rangle\trianglelefteq\langle \act_2, \sigma_2 \rangle$, if $\act_1 = \act_2$ and $\sigma_1 \subseteq \sigma_2$.
\end{definition}
Note that we could have reduced the action labeling to include only the guard matches for concrete instances already.
However, this simplification does not make the branching finite.
Even worse, the resulting LTS would become nondeterministic and looses image-finiteness at the same time. 

As before, the abstraction mechanisms we study are based on greatest lower bounds and least upper bounds of the abstract domain $(\univi{\nullS},\to)$.
Next, we study the intersection abstraction of $C \in \mathbf{2}^{\univ{}}$ with $\alpha_{4}(C)$ in (\ref{eq:alphagamma_null_intersection}).
Generalizing from $(\univi{}, \supseteq)$ we get $(\univi{\nullS},\leftarrow)$ as the less databases $C$ contains, the bigger the abstract instance becomes (cf.\xspace Def.~\ref{def:absint} item 1).
Conversely, $\gamma_{4}(\Inst[I])$ in (\ref{eq:alphagamma_null_intersection}) for abstract instance $\Inst\in\univi{\nullS}$.
\begin{align}
	\alpha_{4}(C) := \bigsqcap C
	&  & 
	\gamma_{4}(\Inst[I]) := \{ \Dnst\in\univ{} \mid \Inst[I]\to\Dnst \} \} \label{eq:alphagamma_null_intersection}
\end{align}

\begin{proposition}
 	$(\alpha_{4},\gamma_{4})$ is a Galois connection.%
\end{proposition}
Using labeled nulls, abstract DMSs using CGs become precise abstractions of their concrete counterparts.
\begin{example}
	We analyze the guard $\guard_{af}$ (does there exist a friendship from $x$ to someone) from example~\ref{ex:guards} on 
	$C = \{
	\{\A,\B,\fAB\},\{\A,\C,$ $\fAC\}
	\}$ and
	$\Inst = \bigsqcap C = \{\A,\fXN{A}{1} \}$.
	In contrast to $\Inst' = \bigcap C = \{\A \}$, we have a friendship with nulls in $\Inst$.
	Now we get homomorphisms $h_{\Inst}:\guard_{af} \to \Inst$ and $h_{\Dnst}:\guard_{af} \to \Dnst$ for each $\Dnst \in C$.
\end{example}
\begin{theorem}\label{theorem:intersection_abstraction_null_forall}
	\REVISED[The Galois connection ]{}$(\alpha_{4},\gamma_{4})$ is an abstract interpretation for CGs\REVISED[ and $\forall$-bisimulation]{}.
\end{theorem}

Generalizing the Galois connection $(\alpha_1,\gamma_1)$ to $\univi{\nullS}$ yields $(\alpha_5,\gamma_5)$ with $\alpha_5=\alpha_1$ and $\gamma_5(\Inst[A]) := \{ \Dnst\in\univ{\nullS} \mid \Dnst\to\Inst[A] \}$.
As for all databases $\Dnst$ without labeled nulls, the existence of a homomorphism from $\Dnst$ to $\Inst[A]$ holds if, and only if, $\Dnst\subseteq\Inst[A]$, the new domain generalizes the original result (i.e., Theorem~\ref{theorem:union_abstraction_forall}) slightly, but without further impact.
After all, labeled nulls are proxies for the existence of constants, whereas CNA guards account for the absence of atoms.

\begin{theorem}\label{theorem:union_abstraction_null_forall}
	\REVISED[The ]{}Galois connection $(\alpha_{5},\gamma_{5})$ with $\alpha_{5}(C) := \bigsqcup C$ and $\gamma_{5}(\Inst[I]) := \{ \Dnst\in\univ{} \mid \Dnst\to\Inst[I] \}$ is an abstract interpretation for CNAs\REVISED[ and $\forall$-bisimilarity]{}.
\end{theorem}

\subsection{Combining Unions and Intersections}

Although the former abstractions already capture existentially quantified variables (i.e., projections), they do not jointly support projections as well as negation.
A corresponding abstraction capturing both is $\alpha_{6}: \mathbf{2}^{\univ{}} \to \univi{\nullS} \times \univi{\nullS} $ with respective concretization $\gamma_{6}: \univi{\nullS} \times \univi{\nullS} \to \mathbf{2}^{\univ{}}$ as defined in \eqref{eq:alphagamma_null_combination}.
The abstract domain is $(\univi{\nullS} \times \univi{\nullS}, \preceq)$. 
$\Inst_1 \preceq  \Inst_2$ is defined as 
$(\Inst^\sqcap_1,\Inst^\sqcup_1) \preceq  (\Inst^\sqcap_2,\Inst^\sqcup_2)$ if and only if $\Inst^ \sqcap_1 \leftarrow \Inst^\sqcap_2$ and $\Inst^\sqcup_1 \to \Inst^\sqcup_2$. The lattice $(\univi{\nullS} \times \univi{\nullS}, \preceq)$ is a combination of the two lattices $(\univi{\nullS}, \leftarrow)$ and $(\univi{\nullS}, \to)$.
\begin{align}
	\alpha_{6}(C) := (\bigsqcap C,\bigsqcup C)
	&  & 
	\gamma_{6}(\Inst) := \{ \Inst^\sqcap \to \Dnst \to \Inst^\sqcup \}\label{eq:alphagamma_null_combination}
\end{align}
\begin{proposition}\label{prop:galois_null_combination}
	$(\alpha_{6},\gamma_{6})$ is a Galois connection.
\end{proposition}

A substitution $\sigma$ holds for a  NCG $\guard$ and an abstract state $\Inst = (\Inst^\sqcap, \Inst^\sqcup)$ if the following holds: $\sigma \in \guard (\Inst)$ if $\sigma \in \guard^+(\Inst^\sqcap)$ and $\sigma \in \guard^-(\Inst^\sqcup)$.
\begin{theorem}\label{theorem:combination_abstraction_null_forall}
	\REVISED[The Galois connection ]{}$(\alpha_{6},\gamma_{6})$ is an abstract interpretation for NCG\REVISED[ and $\forall$-bisimulation]{}.
\end{theorem}	

\begin{example}
With Galois connection $(\alpha_{6},\gamma_{6})$ the guard $\guard_{end} := \exists x,y,z. \fXY$ $\wedge \fYZ \wedge \fZX \wedge \neg P(x,y) \wedge \neg P(y,z) \wedge \neg P(z,x)$ from action $\act_{end}$ (example~\ref{example:semi-reachability}) holds in the abstract and concrete domain.
\end{example}

\section{Related Work}\label{sec:rel_work}
\noindent\textbf{Reasoning about Database-Manipulating Processes.}
Most recent works consider formal process languages for
manipulating relational database in the context of business process modeling~\cite{calvanese2013foundations}.

\emph{Data manipulating systems} (DMS) as considered
in this paper are based on Abdullah et al.~\cite{abdulla2018complexity}.
The authors use the formalism to
study boundaries of decidability of (generally undecidable) reachability of state predicates
in DMS processes.
Their approach employs a formal semantics of DMS processes based on Petri nets and counter machines
in combination with multiset-based abstraction of databases.
Thereupon, Abdullah et al.\xspace impose bounds
on database schemas as well as query evaluation to obtain decidable fragments.
Calvanese et al.~\cite{calvanese2015implementing} also consider a DMS-like language for which they
define an \emph{LTS-based operational semantics} to support CTL model-checking of such systems.
Similar to Abdullah et al., bounds are imposed on the generally infinite state space 
to enable an effective, yet incomplete model-checking procedure.

Cangialosi et al.~\cite{cangialosi2010conjunctive,de2012verification} 
consider a DMS-like formalism called artifact-centric (service) language
to verify process properties expressed in the $\mu$-calculus.
To obtain an effective verification procedure, the authors employ, in accordance to our framework, 
homomorphism equivalence as abstraction and restrict the process language to conjunctive queries, respectively.
Bagheri et al.~\cite{bagheri2011foundations} extend the work of Cangialosi et al.
by supporting negation within first-order queries serving
as preconditions (guards) of transitions.
As a consequence, processes must be restricted to be weakly acyclic 
in order to ensure a finite solution.

Other works use \emph{Petri nets} with data (colored Petri nets) as a DMS-like formalism.
\REVISED{Montali et al.~\cite{montali2017db}} propose DB-nets to integrate data- and process-related aspects of business processes.
\REVISED[In~\cite{montali2019db}, Montali et al.\xspace extended DB-nets by priorities
to guide the selection of steps to enable the usage of practical
tools like CPN tools to validate DB-nets.]{}
In~\cite{montali2016model}, Montali et al.\xspace adopt soundness checks (including reachability) known from workflow nets
to DB-nets, where a finite solution is ensured by employing different notions of boundedness.
This work work \REVISED[has been recently]{has recently been} extended by Ghilardi et al.~\cite{ghilardi2020petri,ghilardi2022petri}
to support conjunctive queries with atomic negation and \REVISED[exists]{existential} quantifiers.

To summarize, most works impose 
bounds on the state space and/or restrictions of guard/query languages
to ensure effective reasoning about semantic properties of DMS-like processes.
However, to the best of our knowledge, 
none of these works provide a comprehensive
decomposition hierarchy of guard/query expressions 
together with a precise characterization
of corresponding semantic-preserving abstractions.

\smallskip
\noindent\textbf{Abstraction Techniques for Databases.}
Halder et al.~\cite{halder2010abstract,halder2012abstract}
apply principles of \emph{abstract interpretation} in a more practical setting to 
define fine-grained abstractions for SQL query expressions.
For approximating query result sets, 
query- and database-specific lattice-based abstractions are applied to 
value ranges of attribute constraints in selection conditions.
In other works, abstract interpretation is mostly
used to formalize the interface between database languages
and programming languages.
Baily et al.~\cite{bailey1999abstract} apply abstract interpretation
for termination analysis for a functional programming language performing
database manipulations.
Similar attempts are proposed by Amato et al.~\cite{amato1993data} and 
Toman et al.~\cite{toman2005constraint} to reason about the interplay 
between imperative programming and
database manipulating operations.
However, using abstract interpretation to characterize
an implementation-independent hierarchy of database  
abstractions as proposed in this paper
has not yet been considered.

Besides abstract interpretation, \emph{symbolic execution}
techniques are also frequently considered to
effectively cope with large/infinite state spaces of database
systems.
In these approaches, sets of databases instances are symbolically represented
using logical constraints, where most recent works employ this approach for test-data generation
from/for databases~\cite{pan2011database,marcozzi2013relational,li2010dynamic,lo2010generating}. 
In contrast, elaborating a hierarchy of \emph{symbolic} abstractions using
different fragments of propositional logics similar to our approach, has not
been investigated so far.

\section{Conclusion}\label{sec:conclusion}

We proposed a hierarchy
of abstract domains for representing (possibly infinite) sets of databases instances
in a final way based on the principles of abstract interpretation.
The resulting hierarchy is semantic-preserving up-to bisimilarity and is shaped by
different fragments of first-order logics serving as 
guard language of database-manipulating processes.
As a future work, our framework can be instantiated in different
ways to facilitate DMS model-checking (e.g., considering corresponding fragments
of the modal $\mu$-calculus as specification language).
\REVISED[We therefore complete our framework to comprise behavioral-preserving abstractions w.r.t. $\exists$-bisimulation.]{}
To this end, a purely abstract step semantics is to be defined 
which allows us to explore the abstract LTS (e.g., starting from 
all possible initial database instances).
We further plan to enrich DMS by a formal process
language like Petri nets and CCS to investigate effects as induced by
constructs like guarded choice and concurrent actions. %

\paragraph{Acknowledgements.}
Stephan Mennicke is partly supported by 
DFG (German Research Foundation) in project 389792660 (TRR 248, \href{https://www.perspicuous-computing.science/}{CPEC}),
by the BMBF (Federal
Ministry of Education and Research) under project 13GW0552B (\href{https://digitalhealth.tu-dresden.de/projects/kimeds/}{KIMEDS}),
in the \href{https://www.scads.de/}{Center for Scalable Data Analytics and
Artificial Intelligence} (ScaDS.AI),
and by BMBF and DAAD (German Academic Exchange Service) in project 57616814 (\href{https://www.secai.org/}{SECAI, School of Embedded and Composite AI}).

 \bibliographystyle{splncs04}
 %

 %

%
\newpage
\appendix
%
\section{Domains are Lattices}
We start by noticing that we have seven different domains, all of which are bounded lattices:
\begin{enumerate}[(1)]
  \item $(\mathbf{2}^{\univ{}}, \subseteq)$: the domain of concrete instances;
  \item $(\univi{}, \subseteq)$;
  \item $(\univi{}, \supseteq)$;
  \item $(\univi{}\times\univi{}, \preceq)$;
  \item $(\univi{\nullS}, \to)$;
  \item $(\univi{\nullS}, \leftarrow)$;
  \item $(\univi{\nullS}\times\univi{}, \preceq)$;
\end{enumerate}

\begin{proposition}
  $(\mathbf{2}^{\univ{}}, \subseteq)$ is a complete, thus, bounded lattice.
\end{proposition}
\begin{proof}
  Every powerset domain is a bounded lattice with infimum $\Dnst \cap \Dnst'$, supremum $\Dnst \cup \Dnst'$, $\top = \univ{} = \bigcup_{\Dnst\in\univ{}} \Dnst$ and $\bot = \emptyset = \bigcap_{\Dnst\in\univ{}} \Dnst$.
  Furthermore, it is a complete lattice, meaning that for every subset $X$ of $\mathbf{2}^{\univ{}}$, greatest lower bound (i.e., $\bigcap_{\Dnst\in X} \Dnst$) and least upper bound (i.e., $\bigcup_{\Dnst\in X} \Dnst$) exist and are unique.\qed
\end{proof}

\begin{proposition}\label{prop:null-hom-lattice}
  $(\univi{\nullS},\to)$ is a complete lattice.
\end{proposition}
\begin{proof}
  Without loss of generality, we assume that instances $\Inst$ and $\Jnst$, if any, use distinct nulls.
  \begin{description}
    \item[Supremum:] For $\Inst,\Jnst\in\univi{\nullS}$, we get the unique least upper bound $\Inst\sqcup\Jnst$ by $\Inst \cup \Jnst$.
    Let $\Inst[U]$ be any other upper bound.
    Then there are homomorphisms $h_1 : \Inst \to \Inst[U]$ and $h_2 : \Jnst \to \Inst[U]$, so that $h_1 \cup h_2$ is a homomorphism $\Inst\sqcup\Jnst \to \Inst[U]$.
    \item[Infimum:] For $\Inst,\Jnst\in\univi{\nullS}$, we find that $\Inst\cap\Jnst$ is a lower bound of $\Inst$ and $\Jnst$.
    (the respective homomorphism is the identity on $\Inst\cap\Jnst$).
    It remains to be shown that there is a greatest lower bound of $\Inst$ and $\Jnst$, denoted $\Inst\sqcap\Jnst$, that is unique.
    To prove it, we use the observation that if $\Inst[L]_1$ and $\Inst[L]_2$ are lower bounds of $\Inst$ and $\Jnst$, then $\Inst[L]_1\sqcup\Inst[L]_2$ is also a lower bound.
    In that case, $\Inst[L]_1\sqcup \Inst[L]_2$ is a lower bound greater than $\Inst[L]_1$ and $\Inst[L]_2$.
    If $\Inst[L]_1$ and $\Inst[L]_2$ are distinct greatest lower bounds, then $\Inst[L]_1\sqcup\Inst[L]_2$ is a lower bound greater than $\Inst[L]_1$ and $\Inst[L]_2$, contradicting the assumption that $\Inst[L]_1$ and $\Inst[L]_2$ are greatest lower bounds.
    Hence, such a pair of distinct greatest lower bounds $\Inst[L]_1$ and $\Inst[L]_2$ must not exist, implying there is a unique greatest lower bound.

    For the remainder of the proof, let us assume that $\Inst[L]_1$ and $\Inst[L]_2$ use distinct labeled nulls.
    As $\Inst[L]_1$ and $\Inst[L]_2$ are lower bounds of $\Inst$ and $\Jnst$, there are homomorphisms $h_1^X : \Inst[L]_1 \to X$ and $h_2^X : \Inst[L]_2 \to X$ for $X\in\{\Inst,\Jnst\}$.
    Then $h_1^X \cup h_2^X$ is a homomorphism certifying for $\Inst[L]_1\sqcup\Inst[L]_2 \to X$, showing that $\Inst[L]_1 \sqcup \Inst[L]_2$ is actually a lower bound of \Inst and \Jnst.
    \item[Completeness:] Let $\Inst[X]\subseteq\univi{\nullS}$.
    Least upper bound $\bigsqcup \Inst[X] = \bigcup \Inst[X]$ is defined and is unique.
    For the greatest lower bound $\bigsqcap\Inst[X]$, we get uniqueness by following similar arguments as in the binary case.
  \end{description}
  Thus, we get $\bigsqcup_{\Inst\in\univi{\nullS}} \Inst = \bigcup_{\Inst\in\univi{\nullS}} \Inst = \univi{\nullS}$ as $\top$ and $\bigsqcap_{\Inst\in\univi{\nullS}} \Inst = \emptyset$ as $\bot$.
  Note that $\bigcap_{\Inst\in\univi{\nullS}} \Inst \subseteq \bigsqcap_{\Inst\in\univi{\nullS}} \Inst$, thus $\bot = \emptyset$.\qed
\end{proof}
\begin{corollary}\label{cor:lattices-rest}
  (1) $(\univi{\nullS},\leftarrow)$, (2) $(\univi{},\subseteq)$, and (3) $(\univi{},\supseteq)$ are complete lattices.
\end{corollary}
\begin{proof}
  \begin{enumerate}[(1)]
    \item Follows from duality of Proposition~\ref{prop:null-hom-lattice}.
    \item Follows as a special case of Proposition~\ref{prop:null-hom-lattice}.
    Therefore note, $\univi{}\subsetneq\univi{\nullS}$ and for each $\Inst[I],\Inst[J]\in\univi{}$, we get $\Inst[I]\subseteq\Inst[J]$ if, and only if, $\Inst[I]\to\Inst[J]$.
    \item Follows as a special case of (1) (cf.\xspace proof of (2)) or by duality of (2).\qed
  \end{enumerate}
\end{proof}

\begin{lemma}\label{lemma:cartesian-lattice}
  If $(A,\leq_A)$ and $(B,\leq_B)$ are complete lattices, then so is $(A\times B, \preceq)$ with for $(a_1, b_1), (a_2, b_2) \in A \times B$, we get $(a_1,b_1)\preceq (a_2,b_2)$ if, and only if, $a_1 \leq_A a_2$ and $b_1 \leq_B b_2$. 
\end{lemma}
\begin{proof}
  As $(A,\leq_A)$ and $(B,\leq_B)$ are complete lattices, every subset $X$ of $A$ (or $B$, resp.) has an infimum $\sqcap_A X$ ($\sqcap_B X$, resp.) and a supremum $\sqcap_A X$ ($\sqcap_B X$, resp.).
  Then for $(a_1,b_1), (a_2, b_2)\in A\times B$, $(a_1 \sqcap_A a_2, b_1 \sqcap_B b_2)$ is the infimum $(a_1,b_1) \sqcap (a_2,b_2)$:
  If $(a_3,b_3)\in A\times B$ is a lower bound of $(a_1,b_1)$ and $(a_2,b_2)$, then $a_3$ is a lower bound $a_1$ and $a_2$, and $b_3$ is a lower bound of $b_1$ and $b_2$ (by definition of $\preceq$).
  Hence, $a_3 \leq_A a_1 \sqcap_A a_2$ and $b_3 \leq_B b_1 \sqcap_B b_2$ since $(A,\leq_A)$ and $(B,\leq_B)$ are complete lattices.
  This means, $(a_3,b_3)\preceq (a_1,b_1)\sqcap (a_2,b_2)$ (again by definition of $\preceq$).
  A similar line of arguments can be taken for the supremum $(a_1,b_1) \sqcup (a_2,b_2)$.

  The bounds of the lattice are $(\top_A,\top_B) = \top$ and $(\bot_A,\bot_B)=\bot$.
  For every subset $Y$ of $A\times B$, let $Y_A := \{ a \mid (a,b)\in Y \}$ and $Y_B := \{ b \mid (a,b)\in Y\}$.
  Then we obtain the infimum of $Y$ by $\bigsqcap Y := (\bigsqcap_A Y_A, \bigsqcap_B Y_B)$ and the supremum by $\bigsqcup Y := (\bigsqcup_A Y_A,\bigsqcup_B Y_B)$, all justified by the fact that the input lattices are complete.\qed
\end{proof}

\begin{corollary}
  $(\univi{}\times\univi{},\preceq)$ and $(\univi{\nullS}\times\univi{\nullS},\preceq)$ are complete lattices.
\end{corollary}
\begin{proof}
  Follows from Lemma~\ref{lemma:cartesian-lattice} together with Proposition~\ref{prop:null-hom-lattice} and Corollary~\ref{cor:lattices-rest}.\qed
\end{proof}
 %
\section{Galois Connections}
Instead of proving Propositions~\ref{theorem:union_galois_connection}--\ref{prop:galois_null_combination}, showing that $(\alpha_1,\gamma_1)$, \ldots, $(\alpha_6,\gamma_6)$ are Galois connections, we extract the common principle from these connections into the following theorem.
\begin{theorem}\label{thm:}
  Let $(\absDom,\sqsubseteq)$ be an abstract domain, such that $\univ{}\subseteq\absDom$ and $(\mathbf{2}^{\univ{}},\subseteq)$ a concrete domain, both complete lattices.
  Then functions $\alpha : \mathbf{2}^{\univ{}} \to \absDom$ with $\alpha(C) := \bigsqcup C$ and $\gamma : \absDom \to \mathbf{2}^{\univ{}}$ with $\gamma(\Inst[A]) = \{ \Dnst \mid \Dnst\in\univ{} \wedge \Dnst \sqsubseteq \Inst[A] \}$ form a Galois connection $(\alpha,\gamma)$.
\end{theorem}
\begin{proof}
  We need to show that $(\alpha,\gamma)$ satisfies the three properties of a Galois connection (cf.\xspace Definition~\ref{def:absint}).
  \begin{description}
    \item[Totality/Monotonicity:]
    Since $(\absDom,\sqsubseteq)$ is a complete lattice, $\alpha$ is total as it uses the abstract least upper bound of a set of databases, given as input.
    Also, $\gamma$ is total since, in the worst case, $\gamma(\Inst[A])=\emptyset$ if there is not database being smaller than $\Inst[A]$ up to $\sqsubseteq$.
    \item[Concretization Preservation:]
    We need to show that for any $C\in\mathbf{2}^{\univ{}}$, $C \subseteq \gamma \circ \alpha(C)$.
    As $\alpha(C)$ is the (abstract) least upper bound of all the databases in $C$, we get $\Dnst \sqsubseteq \alpha(C)$ for $\Dnst\in C$.
    Since $\gamma(\alpha(C))$ is the set of all databases smaller than $\alpha(C)$, we have that $C \subseteq \gamma(\alpha(C))$.
    \item[Abstraction Optimality:] 
    We need to show that for any $\Inst[A]\in\absDom$, $\alpha\circ\gamma(\Inst[A]) \sqsubseteq \Inst[A]$.
    $\gamma(\Inst[A])$ is the set of all databases (abstractly) smaller than $\Inst[A]$.
    Thus, $\Inst[A]$ is an upper bound of $\gamma(\Inst[A])$.
    Since $\alpha(\gamma(\Inst[A]))$ produces the least upper bound of $\gamma(\Inst[A])$, the result follows.\qed
  \end{description}
\end{proof}
Propositions~\ref{theorem:union_galois_connection}--\ref{prop:galois_null_combination} follow as corollaries.
 %
\section{Bisimilarity}
\subsection{Proof of Theorem~\ref{theorem:union_abstraction_forall}}
The Galois connection we are about to prove to be an abstract interpretation framework is $(\alpha_1,\gamma_1)$ with $\alpha_1(C) := \bigcup_{\Dnst\in C} \Dnst$ and $\gamma_1(\Inst[A]) := \{ \Dnst\in\univ{} \mid \Dnst \subseteq \Inst[A] \}$.
The guard language is CNA and we need to give a $\forall$-bisimulation.
 
We show that 
		$$R  := \{ (\alpha(C), C) \mid C \in \mathbf{2}^{\univ{}} \}$$ 
		is a $\forall$-bisimulation.
		Let $(\Inst,C)\in R$ (i.e., $\Inst=\alpha(C)$),
		$\actsigma$ be an action such that $\act = (g,\setDel,\setAdd)$ and $g = \forall \vec y .\ \neg a_{1} \wedge \ldots \wedge \neg a_{m}$ ($m\in\mathbb{N}$) is a CNA.
		\begin{enumerate}
			\item If $\Inst \transrel{\actsigma} \Inst'$ with $\Inst' = (\Inst \setminus \setDel\sigma^{\star}) \cup \setAdd\sigma^{\star}$, we need to show that $C \transrel{\actsigma} C'$ for $C' = \{ \Dnst'\in\univ{} \mid \exists \Dnst\in C : \Dnst \transrel{\actsigma} \Dnst' \}$, such that $\Dnst\enabled{\act,\sigma}$ for all $\Dnst\in C$ and $(\Inst',C')\in R$ (i.e., $\alpha_{1}(C')=\Inst'$).
			
			As $\Inst\enabled{\act,\sigma}$, there is no function $h : \constS \cup \varS \to \constS$ such that $\sigma\subseteq h$ and $h(g)\cap\Inst\neq\emptyset$.
			Let $\Dnst\in C$.
			Suppose there is a function $h'$ such that $\sigma\subseteq h'$ and $h'(g)\cap\Dnst\neq\emptyset$.
			Then $h'(g) \cap \Inst \neq \emptyset$ as $\Dnst\subseteq\Inst$ (by $\alpha_1$), contradicting the assumption that no such function exists.
			Thus, $\Dnst\enabled{\act,\sigma}$ which holds for arbitrary $\Dnst\in C$.
			Moreover, for every $\Dnst\in C$ there is a $\Dnst'$, such that $\Dnst \transrel{\actsigma} \Dnst'$ and $\Dnst' = (\Dnst \setminus \setDel\sigma^{\star}) \cup \setAdd\sigma^{\star}$.
			The collection of all such $\Dnst'$ forms the set $C'$, such that $C \transrel{\actsigma} C'$.
			Finally, we get
      \begin{equation}\label{eq:alpha2inst}
			\begin{array}{rcl}
				\alpha_{1}(C') &=& \bigcup \{ (\Dnst \setminus \setDel\sigma^{\star}) \cup \setAdd\sigma^{\star} \mid \Dnst\in C \} \\
											 &=& (\bigcup \{ \Dnst \mid \Dnst\in C \} \setminus \setDel\sigma^{\star}) \cup \setAdd\sigma^{\star} \\
											 &=& (\alpha_{1}(C) \setminus \setDel\sigma^{\star}) \cup \setAdd\sigma^{\star} \\
											 &=& (\Inst \setminus \setDel\sigma^{\star}) \cup \setAdd\sigma^{\star} = \Inst'\text,
			\end{array}
    \end{equation}
			proving the fact that $(\Inst',C')\in R$.

			\item If $C \transrel{\actsigma} C'$, then for every $\Dnst\in C$, $\Dnst\transrel{\actsigma} \Dnst'$ such that $\Dnst' = (\Dnst \setminus \setDel\sigma^{\star}) \cup \setAdd\sigma^{\star}$ and $\Dnst'\in C'$.
			It holds that there is no function $h : \constS \cup \varS \to \constS$, such that $\sigma\subseteq h$ and $h(g)\cap \Dnst \neq\emptyset$ for every $\Dnst\in C$.
			Suppose, there is an $h'$ such that $\sigma\subseteq h'$ and $h'(g)\cap\Inst\neq\emptyset$.
			Then there is an atom $a\in h'(g)\cap\Inst$ that is also included in some $\Dnst\in C$ (as $\Inst = \bigcup C$).
			Hence, $h'(g)$ has a non-empty intersection with that $\Dnst$, contradicting the assumption that $\Dnst\enabled{\act, \sigma}$ for every $\Dnst\in C$.
			Therefore, $\Inst\enabled{\act,\sigma}$ and $\Inst\transrel{\actsigma} (\Inst \setminus \setDel\sigma^{\star}) \cup \setAdd\sigma^{\star} = \Inst'$.
			By the same lines as above (cf.\xspace \eqref{eq:alpha2inst}), $\alpha_{1}(C')=\Inst'$ implying $(\Inst',C')\in R$.\qed

		\end{enumerate}

\subsection{Remark on Theorem~\ref{theorem:union_abstraction_null_forall}}
Note, although the abstract domain is slightly different, the abstraction function still uses set union as abstraction mechanism.
Therefore, and since the same guard language is conjectured, the theorem directly follows for the same reasons as given in the proof of Theorem~\ref{theorem:union_abstraction_forall}.

\subsection{Proof of Theorem~\ref{theorem:intersection_abstraction_null_forall}}
Let $(\Inst,C)\in R$ (i.e., $\Inst=\alpha(C)$).
		For some action $\actsigma$ with $\act = (\guard,\setDel,\setAdd)$,
		\begin{enumerate}
			\item $\Inst \transrel{\actsigma} \Inst'$ with $\Inst' = (\Inst \setminus \setDel\sigma^{\star}) \cup \setAdd\sigma^{\star}$ we need to show that $C \transrel{\actsigma} C'$, such that for all $\Inst[D]\in C$, there is a $\Inst[D]'\in C'$, such that $\Inst[D] \transrel{\actsigma} \Inst[D]'$ and $(\Inst',C')\in R$.
			
			Let $\Dnst\in C$.
			Since $\Inst=\bigsqcap C$ by $\alpha_4$, there is a homomorphism $id_{\Inst}:\Inst\to\Dnst$.
			From Prop.~\ref{prop:hom-cg-lemma} a homomorphism $h_\guard:\guard\rightarrow\Inst$ exists.
			Then $h = h_\guard \circ h_{\Inst}$ is a homomorphism $\guard \to \Dnst$ and $\sigma\in \guard(\Dnst)$ follows.
			Hence, $\Dnst\transrel{\actsigma} (\Dnst \setminus \setDel\sigma^{\star}) \cup \setAdd\sigma^{\star}$.
			This argument holds for all $\Dnst\in C$.		
			Thus $C \transrel{\actsigma} C' = \{ (\Dnst \setminus \setDel \sigma^*) \cup \setAdd\sigma^{\star} \mid \Dnst\in C \}$.
			
			By construction of $\Inst'$ and all $\Inst[D]'\in C'$, it is clear that $\Inst'$ is a lower bound of each $\Inst[D]'\in C'$.
			Thus, $\Inst'\to\alpha_4(C')=\bigsqcap C'$.
			It follows  $\Inst' \homeq \alpha_4(C')$ and, thus, $(\Inst',C')\in R$.
			\item $C \transrel{\actsigma} C'$ such that for every $\Dnst\in C$, $\sigma\in \guard(\Dnst)$, $\Dnst\transrel{\actsigma}\Dnst'$ and $\Dnst'\in C'$.
			It follows $\sigma\in \guard(\sqcap{C})$ and so $\sigma\in \guard(\Inst)$.
			Construction of $\Inst'$ according to the definition. Similarities pull out. Thus, $\alpha(C') \homeq \Inst'$.\qed
		\end{enumerate}

\subsection{Remark on Theorem~\ref{theorem:intersection_abstraction_forall}}
\label{app:proof_intersection}
Compared to Theorem~\ref{theorem:intersection_abstraction_null_forall}, projection-free conjunctive guards constitute the guard language for this theorem.
The reason may be found in the steps considered on the concrete domain (i.e., item 2 in the proof of Theorem~\ref{theorem:intersection_abstraction_null_forall}).
Consider the two databases
\begin{align*}
  \Dnst_1= & \{ p(A), p(B), f(A,B) \} & \Dnst_2= & \{ p(A), p(C), f(A,C) \}\text.
\end{align*}
Then the set intersection (also known as abstraction $\alpha_2$) is
\begin{align*}
  \Inst[A] = \Dnst_1 \cap \Dnst_2= & \{ p(A) \}\text,
\end{align*}
while the abstract intersection in domain $(\univi{\nullS},\leftarrow)$ is
\begin{align*}
  \Inst[B] = \Dnst_1 \sqcap \Dnst_2= & \{ p(A), f(A,n) \}\text.
\end{align*}
In $\Inst[B]$ the existence of a value $x$ such that $f(A,x)$ is expressed by the labeled null $n$.
Thus, a guard asking for the existence of a value $x$, such that $f(A,x)$ can be matched on $\Inst[B]$.
On $\Inst[A]$, on the other hand, the same guard has no match, although both concrete databases $\Dnst_1$ and $\Dnst_2$ list one $f$-atom accounting for the connection to $A$.
Thus, projections are expressed inside the abstract instances using labeled nulls (e.g., $\Inst[B]$), but is not present in the simpler set-based abstractions (e.g., $\Inst[A]$).

The proof of Theorem~\ref{theorem:intersection_abstraction_forall} follows the same lines as the proof of Theorem~\ref{theorem:intersection_abstraction_null_forall}, using $\subseteq$ instead of postulating the existence of homomorphisms\footnote{Note, $\subseteq$ ultimately entails the existence of a homomorphism, the identity function.}.
In item 2, $\sigma$ can only be considered a match on the abstract instance $\Inst$ if the conjunctive guard $g$ is projection-free (i.e., does not use existential quantification).

\subsection{Proof of Theorem~\ref{theorem:combination_abstraction_null_forall}}
Here we consider Galois connection $(\alpha_6,\gamma_6)$ on the abstract domain $(\univi{\nullS}\times\univi{\nullS}, \preceq)$ with $\alpha_6(C) := (\bigsqcap C, \bigsqcup C)$ and for $\Inst \in \univi{\nullS}\times\univi{\nullS}$, we get if $\Inst = (\Inst_1,\Inst_2)$, then $\Inst^\sqcap = \Inst_1$ and $\Inst^\sqcup = \Inst_2$.
We prove $\forall$-bisimilarity for DMS using safe normal conjunctive guards.

Towards this goal, we show that
$$R = \{ (\alpha_6(C), C) \mid C\subseteq\univ{} \}$$
is a $\forall$-bisimulation.
For $(\Inst,C)\in R$ and $\actsigma$,
\begin{enumerate}
	\item if $\Inst \transrel{\actsigma} \Inst'$,
	we need to show that $C\transrel{\actsigma} C'$ such that (a) for all $\Dnst\in C$, $\Dnst \transrel{\actsigma} \Dnst'$ and $\Dnst'\in C$, (b) for each $\Dnst'\in C$ there is a $\Dnst\in C$ with $\Dnst \transrel{\actsigma} \Dnst'$, (c) $(\Inst',C')\in R$.
	Recall that $\Inst = ()$ and \act uses an NCG guard $g$, such that $\sigma\in g^+(\Inst^\sqcap)$ and $\sigma\in g^-(\Inst^\sqcup)$.
	For $\Dnst\in C$, we have that $\Inst^\sqcap \to \Dnst$ and $\Dnst \to \Inst^\sqcup$.
	By the argumentations in theorem \ref{theorem:intersection_abstraction_null_forall} and \ref{theorem:union_abstraction_null_forall}, we obtain $\sigma\in g^+(\Dnst) \cup g^-(\Dnst)$.
	Thus, $\Dnst \transrel{\actsigma} \Dnst'$ for all $\Dnst\in C$ and $\Dnst'\in C'$ such that $\Dnst' = (\Dnst \setminus \setDel\sigma^\bullet) \cup \setAdd\sigma^\bullet$ for some extension $\sigma^\bullet$ of $\sigma$ (surely, $\sigma$ and $\sigma^\bullet$ are compatible as $\sigma^\bullet$ extends $\sigma$).
	Following the arguments of the previous theorems separately for $\Inst'^\sqcap$ and $\Inst'^\sqcup$, we obtain the result that $\Inst' = \alpha(C')$, implying $(\alpha(C'),\Inst')\in R$.
	\item if $C \transrel{\actsigma} C'$, then the same separation of concerns for $\alpha(C)^\sqcap$ and $\alpha(C)^\sqcup$ applies.
	The necessary arguments can be found in the proofs of theorems \ref{theorem:intersection_abstraction_null_forall} and \ref{theorem:union_abstraction_null_forall}.\qed
\end{enumerate}

\subsection{Remark on Theorem~\ref{theorem:combination_abstraction_forall}}
Once again, the proof is similar to the proof of the abstraction on the more general domain.
Since the abstraction function here is composed of the $\alpha_1$ and $\alpha_2$, we have to stick to projection-freeness of the guards once more.
Otherwise, the proof follows the same lines as the proof of Theorem~\ref{theorem:combination_abstraction_null_forall}, respecting the notes mentioned in Sect.~\ref{app:proof_intersection} 
 
%

\end{document}